\documentclass[preprint]{elsarticle}

\usepackage{amsmath}
\usepackage{amssymb} 
\usepackage{pst-node}

\newcommand{\Nu}{\Upsilon}
\newcommand{\nat}{\mbox{$\mathbb{N}$}}
\newcommand{\supp}{\mathit{supp}}
\newcommand{\too}[1]{{\overset{#1}{\to}}}
\newcommand{\pre}{\mathit{pre}}
\newcommand{\Pred}{\mathit{Pred}}
\newcommand{\post}{\mathit{post}}
\newcommand{\Var}{\mathit{Var}}
\newcommand{\Id}{\mathit{Id}}
\newcommand{\minn}{\mathit{min}}
\newcommand{\mMS}[1]{#1^\oplus}
\newcommand{\uc}{\uparrow}
\newcommand{\uca}{\uc_\alpha}
\newcommand{\oM}{\overline{M}}

\newdefinition{definition}{Definition}
\newtheorem{lemma}{Lemma}
\newtheorem{proposition}{Proposition}

\newproof{proof}{Proof}
\newdefinition{example}{Example}

\begin{document}

\begin{frontmatter}

\title{Decision Problems for Petri Nets with Names
\footnote{Authors partially supported by 
the Spanish projects DESAFIOS10 TIN2009-14599-C03-01, PROMETIDOS S2009/TIC-1465 and TESIS TIN2009-14312-C02-01.}
}
\author{Fernando Rosa-Velardo}
\ead{fernandorosa@sip.ucm.es}
\author{David de Frutos-Escrig}
\ead{defrutos@sip.ucm.es}
\address{Dpto. de Sistemas Inform\'aticos y Computaci\'on\\
  Universidad Complutense de Madrid}

\begin{abstract}
  We prove several decidability and undecidability results for $\nu$-PN, an extension
  of P/T nets with pure name creation and name management. We give a simple
  proof of undecidability of reachability, by reducing reachability in nets with inhibitor
  arcs to it. Thus, the expressive power
  of $\nu$-PN strictly surpasses that of P/T nets. We prove that
  $\nu$-PN are Well Structured Transition Systems. In particular, we 
  obtain decidability of coverability and termination, so that the expressive power of Turing machines
  is not reached. Moreover, they are strictly Well Structured, so that the
  boundedness problem 
  is also decidable. We consider two properties,
  width-boundedness and depth-boundedness, that factorize boundedness. 
  Width-boundedness has already been proved to be decidable.
  We prove here undecidability of depth-boundedness. Finally, we obtain Ackermann-hardness
  results for all our decidable decision problems.
\end{abstract}


\end{frontmatter}

\section{Introduction}

Pure names are identifiers with no relation between them
other than equality~\cite{Gordon00}. 
Dynamic name generation has been thoroughly studied, mainly
in the field of security and mobility~\cite{Gordon00}
 because they can
be used to represent channels, as in \hbox{$\pi$-calculus}~\cite{MilnerPW92}, ciphering keys,
as in spi-Calculus~\cite{AbadiG97} or
computing boundaries, as in the Ambient Calculus~\cite{cardelli98mobile}.

In previous works we have studied a very simple extension of P/T nets~\cite{Desel98},
that we called $\nu$-PN~\cite{Secco,fi08}, for name creation and
management.\footnote{Actually, we used the term $\nu$-APN, where the A stands for \emph{Abstract},
though we prefer to use this simpler acronym.} Tokens
in $\nu$-PN are pure names, that can
be created fresh, moved along the net and be used to restrict the firing
of transitions with name matching. 
They essentially correspond to the minimal
OO-nets of~\cite{Kummer00}, where names are used to identify objects.

In this paper we prove several (un)decidability and complexity results
for some decision problems in $\nu$-PN. 
In~\cite{Kummer00} the author proved that reachability is undecidable
for minimal OO-nets, thus
proving that the model surpasses the expressive power of P/T
nets. The same result was obtained independently in~\cite{Secco} for $\nu$-PN.
Both undecidability proofs rely on a weak simulation of a Minsky
machine that preserves reachability.
We present here an alternative and simpler proof of the
same result, based on a simulation of Petri nets with
inhibitor nets (thus, with a much smaller representation gap) that
reduces reachability in the latter (which is undecidable) to reachability 
in $\nu$-PN.

In~\cite{Secco} we proved well structuredness~\cite{AbdullaCJT00,Finkel01Everywhere} 
of a class of nets
we called MSPN. It is easy to see that $\nu$-PN can easily encode MSPN.
We present here full details of the proof of well structuredness
for $\nu$-PN instead of for MSPN, since the former is a much more cleaner
formalism. This gives us decidability of coverability (which is an important
property, since safety properties can be specified in terms of it) and termination~\cite{AbdullaCJT00,Finkel01Everywhere}. 
We also prove that the well structuredness of $\nu$-PN is strict,
so that boundedness (whether there are infinitely many reachable markings) is also decidable~\cite{Finkel01Everywhere}.
Moreover, we work with an extended version of $\nu$-PN,
in which we allow weights in arcs, simultaneous creation of several
fresh names and checks for inequality.

$\nu$-PN can represent infinite state systems that
can grow in two orthogonal directions:
On the one hand, markings may have an unbounded number of different names;
On the other hand, each name may appear in markings an unbounded number of times.
In the first case we will say the net is width-unbounded, and in the second we
will say it is depth-bounded.
In~\cite{atpn10} we proved decidability of width-boundedness by performing
a forward analysis that, though incomplete in general for the computation
of the cover, can decide width-boundedness. In particular,
we instantiated the general framework developed in~\cite{forwardI,forwardII}
for forward analyses of WSTS in the case of $\nu$-PN.

Here we prove undecidability of depth-boundedness. Thus, though
both boundedness concepts are closely related, they behave very differently.
The proof reduces boundedness in reset nets, which is known
to be undecidable~\cite{dufourd98reset}, to depth-boundedness in $\nu$-PN.
This result can be rather surprising. Actually,
the paper~\cite{Dietze07} erroneously establishes the decidability of
depth-boundedness (called t-boundedness there).

\noindent\textbf{Related work. }
Another model based on Petri nets that has names as tokens are
\emph{Data Nets}, which are also WSTS~\cite{Lazic08}. 
In Data Nets, tokens are not pure
in general, but taken from a linearly-ordered infinite domain. Names 
can be created, but they can only be guaranteed to be fresh by
explicitely using the order in the data domain, by taking a datum which
is greater than any other that has been used. Thus, in an
unordered version of Data Nets, names cannot be guaranteed to be
fresh.

Other similar models include Object Nets~\cite{Valk95,Valk98},
that follow the so called nets-within-nets paradigm. In Object
Nets, tokens can themselves be Petri nets that synchronize with
the net in which it lies. This model is supported by the RENEW 
tool~\cite{RENEW}, a tool for the edition and simulation
of Object Petri Nets. Moreover, the RENEW tool can represent $\nu$-PN and, 
therefore, be used to simulate them.

Several papers study the expressive power of Object Nets. 
The paper~\cite{Kohler07} considers a two level restriction of Object Nets,
called Elementary Object Nets (EON), and proves undecidability of reachability for
them. This result extends those in~\cite{KohlerRolke04}.
Moreover, some subclasses are proved to have decidable reachability.
In~\cite{Kohler09} it is shown that, when the synchronization mechanism
is extended so that object tokens can be communicated, then Turing completeness
is obtained. However, in all these models processes (object nets) do not
have identities.

Nested Petri Nets~\cite{Lomazova00} also have nets as tokens, 
that can evolve autonomously, move along the 
system net, synchronize with each other 
or synchronize with the system net (vertical synchronization steps). 
Nested nets are more expressive than $\nu$-PN.
Indeed, it is possible to simulate every $\nu$-PN by means of
a Nested Petri Net which uses only object-autonomous and
horizontal synchronization steps. In Nested Petri Nets, 
reachability and boundedness are undecidable, although
other problems, like termination, remain decidable~\cite{LomazovaSch00}.
Thus, decidability of termination can also be obtained as a 
consequence of~\cite{LomazovaSch00}. Here we obtain decidability
of termination on the way of the proof of 
decidability for boundedness and coverability.

\noindent\textbf{Outline. }
The rest of the paper is structured as follows. Section~\ref{sec:basic}
presents some basic results and notations we will use throughout the
paper. Section~\ref{sec:names} defines $\nu$-PN. Section~\ref{sec:reach} proves undecidability of reachability.
In Sect.~\ref{sec:coverability} we prove decidability of coverability, termination and boundedness,
and we give non-primitive recursive lower bounds for their decision procedures.
Section~\ref{sec:bounded} presents
further results about boundedness and in Section~\ref{sec:conclusion} we present our conclusions.

\section{Preliminaries}\label{sec:basic}

\noindent\textbf{Multisets. }
Given an arbitrary set $A$, we will denote by $\mMS{A}$ the set
of finite multisets of $A$, that is, the set of mappings
$m:A\to\nat$. We will identify each set with the multiset given
by its characteristic function, and use set notation for multisets when
convenient. We denote by $\supp(m)$ the support of $m$, that is, the
set $\{a\in A\mid m(a)>0\}$ and by $|m|=\underset{a\in\supp(m)}{\sum}m(a)$
the cardinality of $m$. Given two multisets $m_1,m_2\in\mMS{A}$ we
denote by $m_1+m_2$ and $m_1\sqcup m_2$ the multisets defined by $(m_1+m_2)(a)=m_1(a)+m_2(a)$
and $(m_1\sqcup m_2)(a)=max\{m_1(a),m_2(a)\}$, respectively.
We will write $m_1\subseteq m_2$ if $m_1(a)\leq m_2(a)$ for every $a\in A$.
In this case, we can define $m_2-m_1$, given by $(m_2-m_1)(a)=m_2(a)-m_1(a)$.
We will denote by $\sum$ the extended multiset sum operator and
by $\emptyset\in\mMS A$ the multiset $\emptyset(a)=0$, for every $a\in A$.
If $f:A\to B$ and $m\in\mMS A$, then we define $f(m)\in\mMS B$ by
$f(m)(b)=\underset{f(a)=b}{\sum}m(a)$. 
Every partial order $\leq$ defined over $A$
induces a partial order $\sqsubseteq$ in the set $\mMS A$, given by 
 $\{a_1,\ldots,a_n\}\sqsubseteq\{b_1,\ldots,b_m\}$ if there is $\iota:\{1,\ldots,n\}\to\{1,\ldots,m\}$ 
injective such that $a_i\leq b_{\iota(i)}$ for all $i$.
We will write $\sqsubseteq_\iota$
to  stress out the use of the mapping $\iota$.\\

\noindent\textbf{wqo. }
A quasi order is a reflexive and transitive binary relation on a set $A$. 
A partial order is an antisymmetric quasi order. 
A quasi order $\leq$ is decidable if
for every $a,b\in A$ we can effectively decide if $a\leq b$. All the quasi orders in this paper are trivially
decidable. For a quasi order $\leq$ we write $a<b$ if $a\leq b$ and $b\not\leq a$. 
A set $B\subseteq A$ is said to be a minor set of $A$ if it does
not contain comparable elements and for all $a\in A$ there is $b\in B$ such that
$b\leq a$. We will write $\minn(A)$ to denote a minor set of $A$. The upward closure
of a subset $B$ is $\uparrow B=\{a\in A\mid\exists b\in B \text{~st~} a\leq b\}$. A subset $B$ is
upward closed iff $B=\uparrow B$. A quasi order is 
well (wqo)~\cite{Milner85} if for every infinite sequence $a_0,a_1,\ldots$ there are
$i$ and $j$ with $i<j$ such that $a_i\leq a_j$. In a wqo $\minn(B)$ is always finite.\\

\noindent\textbf{Transition systems. }
A transition system is a tuple $(S,\to,s_0)$, where $S$ is a (possibly infinite) set 
of states, $s_0\in S$ is the initial state and $\to\subseteq S\times S$. We denote by
$\to^*$ the reflexive and
transitive closure of $\to$. Given $S'\subseteq S$ we denote by 
$\Pred(S')$ the set $\{s\in S\mid s\to s'\in S'\}$.

The reachability
problem in a transition system consists in deciding for a given states $s_f$
whether $s_0\to^* s_f$.
The termination problem consists in deciding whether there is an infinite
sequence $s_0\to s_1\to s_2\to\cdots$. 
The boundedness problem consists in deciding
whether the set of reachable states is finite.
For any transition system $(S,\to,s_0)$ endowed with a quasi order $\leq$ we can define the 
coverability problem, that consists in deciding, given a state $s_f$,
whether there is $s\in S$ reachable such that $s_f\leq s$.\\

\noindent\textbf{WSTS. }
A Well Structured Transition System (WSTS) is a tuple $(S,\to,s_0,\leq)$, where
$(S,\to,s_0)$ is a transition system, $\leq$ is a decidable wqo 
compatible with $\to$ (meaning that $s_1'\geq s_1\to s_2$ implies that there is $s'_2\geq s_2$
with $s'_1\to s'_2$), and so that for all $s\in S$ we can
compute $\minn(\Pred({\uparrow s}))$.
We will refer to these properties as monotonicity of $\to$ with
respect to $\leq$, and effective $Pred$-basis, respectively.\footnote{
Strictly speaking, decidability of the wqo and effective $Pred$-basis 
are not part of the definition of WSTS, but of the so called \emph{effective} WSTS. These properties are needed to ensure
decidability of coverability.}
For WSTS, the coverability and the termination problems are decidable~\cite{AbdullaCJT00,Finkel01Everywhere}.
A WSTS is said to be strict if it satisfies the following strict compatibility
condition: $s_1'>s_1\to s_2$ implies that there is $s'_2> s_2$
with $s'_1\to s'_2$. For strict WSTS, also the boundedness problem is decidable~\cite{Finkel01Everywhere}.\\

\noindent\textbf{Petri Nets. }
Next we define P/T nets in order to set our notations.
A P/T net is a tuple $N=(P,T,F)$ where
$P$ and $T$ are disjoint finite sets of places and transitions,
  respectively,
  and $F:(P\times T)\cup(T\times P)\to\nat$.
A marking $M$ of $N$ is a finite multiset of places of $N$, that
is, $M\in\mMS P$.

As usual, we denote by $t^\bullet$ and $^\bullet t$ the multisets of
postconditions and preconditions of $t$, respectively, that is,
$t^\bullet(p)=F(t,p)$ and $^\bullet t(p)=F(p,t)$.
A transition
$t$ is enabled at
marking $M$ if ${^\bullet t}\subseteq M$.
The reached state of $N$ after the firing of $t$ is
$M'=(M-{^\bullet t})+t^\bullet$.

We will write $M\too{t}M'$ if $M'$ is the reached marking after the firing of $t$
at marking $M$. We also write $M\to M'$ if there is some $t$ such that
$M\too{t}M'$. The reflexive and transitive closure of $\to$ is denoted by
$\to^*$. For a transition sequence $\tau=t_1\ldots t_m$ we will write
$M\too{\tau}M'$ to denote the consecutive firing of transitions $t_1$ to $t_m$, as expected.

\section{Petri nets with name creation}\label{sec:names}

Let us now extend P/T nets with the capability
of name management by defining $\nu$-PN. In a $\nu$-PN
names can be created, communicated  
and matched. We can use this
mechanism to deal with authentication issues~\cite{Secco},
correlation or instance isolation~\cite{DeckerW08}. We
formalize name management by replacing ordinary tokens by
distinguishable ones, thus adding colours to our nets.
We fix a set $\Id$ of names, that can be
carried by tokens of any $\nu$-PN. 
In order to handle these colors, we need matching variables
labelling the arcs of the nets, taken from a fixed set $\Var$. Moreover,
we add a primitive capable of creating new names, formalized by
means of special variables in a set $\Nu\subset\Var$, ranged by $\nu,\nu_1,\ldots$ that can only be instantiated
to fresh names.

\begin{definition}
  A $\nu$-PN is a tuple $N=(P,T,F)$, where $P$ and $T$ are finite
  disjoint sets, $F:(P\times T)\cup (T\times P)\to\mMS\Var$ is 
   such that for every $t\in T$, $\Nu\cap\pre(t)=\emptyset$ and $\post(t)\setminus\Nu\subseteq\pre(t)$, where
$\pre(t)=\bigcup_{p\in P}\supp(F(p,t))$ and $\post(t)=\bigcup_{p\in P}\supp(F(t,p))$.
\end{definition}

We also take $\Var(t)=\pre(t)\cup\post(t)$.
To avoid tedious definitions, along the paper we will consider a fixed $\nu$-PN $N=(P,T,F)$.

\begin{definition}
  A marking of $N$ is a function
  \hbox{$M:P\to\mMS\Id$}. We denote by $\Id(M)$ the set of names in $M$, that is,
$\Id(M)=\underset{p\in P}\bigcup\supp(M(p))$.
\end{definition}

We will assume a fixed initial marking $M_0$ of $N$. Like in other classes of high-order nets, transitions are fired
with respect to a mode, that chooses which tokens are taken
from preconditions and which are put in postconditions.
Given a transition $t$ of a net $N$, a mode of $t$ is an
injection $\sigma:\Var(t)\to\Id$, that instantiates each
variable to an identifier. We will use $\sigma,\sigma',\sigma_1\ldots$
to range over modes.

\begin{definition}
 Let $M$ be a marking, $t\in T$ 
 and $\sigma$ a mode for $t$.
 We say $t$ is enabled with mode
 $\sigma$ if
 for all $p\in P$, $\sigma(F(p,t))\subseteq M(p)$ and $\sigma(\nu)\notin\Id(M)$ for all $\nu\in\Nu\cap\Var(t)$. 
The reached state after the firing of $t$ with mode $\sigma$ is the marking
 $M'$, given by $$M'(p)=(M(p)-\sigma(F(p,t)))+\sigma(F(t,p))~~\mathit{for~all~}p\in P.$$
\end{definition}

We will write $M\too{t(\sigma)}M'$ to
denote that $M'$ is reached from $M$ when $t$ is fired with mode $\sigma$,
and extend the notation as done for P/T nets. In particular,
for a sequence $\tau=t_1(\sigma_1)\ldots t_m(\sigma_m)$ we will write
$M\too{\tau}M'$ to denote the consecutive firings of $t_1(\sigma_1)$ to
$t_m(\sigma_m)$. We will denote by $Reach(N)$ the set of reachable markings
of $N$. Finally, we will assume that $\bullet\in\Id$, so that we can also
have ordinary tokens in our nets.

Figure~\ref{fig:aut} depicts a simple $\nu$-PN with four places and a single
transition. This transition moves one token from $p_1$ to $q_1$ (because of variable
$x$ labelling both arcs), removes a token from $p_1$ and $p_2$ provided they
carry the same name (variable $y$ appears in both incoming arcs but it does not appear
in any outgoing arc), and two different names are created, one appears both
in $q_1$ and $q_2$ (because of variable $\nu_1\in\Nu$) and the other appears
only in $q_2$ (because of variable $\nu_2\in\Nu$).

\begin{figure}[!t]
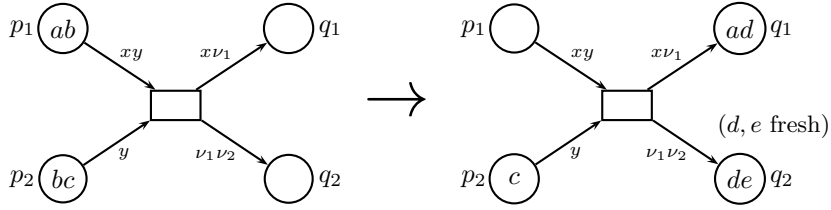

\vspace{1.3cm}
\begin{center}
 \scalebox{1}{
        \cnodeput(-4.5,1){p0}{$ab$}
        \cnodeput(-1.5,1){p1}{\white$ab$}
        \rput[B](-3,-0.1){\Rnode{ta}{\psframebox{\white$~~a$}}}
        \cnodeput(-4.5,-1){p2}{$bc$}
        \cnodeput(-1.5,-1){p3}{\white$be$}
        \uput{0}[0](-5.2,1){$p_1$}
        \uput{0}[0](-1.1,1){$q_1$}
        \uput{0}[0](-5.2,-1){$p_2$}
        \uput{0}[0](-1.1,-1){$q_2$}
        \ncline[nodesep=0pt]{->}{p0}{ta}
        \Aput[1pt]{\scriptsize{$xy$}}
        \ncline[nodesep=0pt]{->}{ta}{p1}
        \Aput[1pt]{\scriptsize{$x\nu_1$}}
        \ncline[nodesep=0pt]{->}{p2}{ta}
        \Bput[1pt]{\scriptsize{$y$}}
        \ncline[nodesep=0pt]{->}{ta}{p3}
        \Bput[1pt]{\scriptsize{$\nu_1\nu_2$}}
        \uput{0}[0](-0.5,0){\Huge$\rightarrow$}
        \cnodeput(1.5,1){p0}{\white$aa$}
        \cnodeput(4.5,1){p1}{$ad$}
        \rput[B](3,-0.1){\Rnode{ta}{\psframebox{\white$~~a$}}}
        \cnodeput(1.5,-1){p2}{$~c~$}
        \cnodeput(4.5,-1){p3}{$de$}
        \uput{0}[0](4.2,-0.3){\small{($d,e$ fresh)}}
        \uput{0}[0](0.8,1){$p_1$}
        \uput{0}[0](4.9,1){$q_1$}
        \uput{0}[0](0.8,-1){$p_2$}
        \uput{0}[0](4.9,-1){$q_2$}
        \ncline[nodesep=0pt]{->}{p0}{ta}
        \Aput[1pt]{\scriptsize{$xy$}}
        \ncline[nodesep=0pt]{->}{ta}{p1}
        \Aput[1pt]{\scriptsize{$x\nu_1$}}
        \ncline[nodesep=0pt]{->}{p2}{ta}
        \Bput[1pt]{\scriptsize{$y$}}
        \ncline[nodesep=0pt]{->}{ta}{p3}
        \Bput[1pt]{\scriptsize{$\nu_1\nu_2$}}
 }\end{center}
\vspace{0.8cm} \caption{A simple $\nu$-PN}\label{fig:aut} \vspace{-0.4cm}
\end{figure}

Notice that we demand modes to be injections (unlike in~\cite{fi08}), which
formalizes the fact that we can check for inequality. For instance, in the example
in Fig.~\ref{fig:aut} the two tokens taken from $p_1$ must carry
different names because we are labelling the arc from $p_1$ to $t$ with two
different variables, namely $x$ and $y$.
The capability of checking for inequality among all the names involved in the firing of a transition 
improves the expressive power of the model (see Fig.~\ref{fig:igualVsNoIgual}).
The problem of proving that this improvement is strict is still open.

If a $\nu$-PN has no arc labelled with 
variables from $\Nu$ then only a finite number of identifiers (those in the
initial marking) can appear in any reachable marking. It is easy to see that these 
nets can be expanded to an equivalent P/T net. In particular, reachability is decidable for 
any such net, as it is for P/T nets~\cite{Esparza94}, unlike for \hbox{$\nu$-PN~\cite{Kummer00}.}

We will work with a subclass of $\nu$-PN without weights and in which transtions
can at most create one fresh name.

\begin{definition}\label{def:normal}
 A $\nu$-PN $N=(P,T,F)$ is normal if there is $\nu\in\Nu$ such that:
 \begin{itemize}
  \item for every pair $(x,y)\in (P\cup T)\times (T\cup P)$, $|F(x,y)|\leq 1$,
  \item if $F(x,y)\cap\Nu\neq\emptyset$ then $F(x,y)=\{\nu\}$.
 \end{itemize}
\end{definition}

Every $\nu$-PN can be simulated by a normal $\nu$-PN.
Intuitively, the simulation considers for each transition several transitions that must be fired
consecutively, whenever the original net takes several tokens from
the same place. Since the firing of a transition in the original net
becomes non-atomic in the simulation, it can introduce deadlocks
(whenever the ``transaction'' cannot be accomplished).
However, it preserves all the properties we will consider in this paper. 
Therefore, from now on we will assume that $\nu$-PNs are normal when needed.

\begin{figure}[!t]
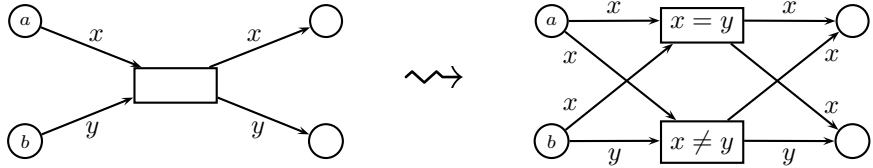

\vspace{0.8cm}
\begin{center}
 \scalebox{1}{
        \cnodeput(-6,0.8){p0}{\scriptsize$a$}
        \cnodeput(-2,0.8){p1}{\white\scriptsize$a$}
        \rput[B](-4,-0.1){\Rnode{ta}{\psframebox{\white$x=y$}}}
        \cnodeput(-6,-0.8){p2}{\scriptsize$b$}
        \cnodeput(-2,-0.8){p3}{\white\scriptsize$b$}
        \ncline[nodesep=0pt]{->}{p0}{ta}
        \Aput[1pt]{$x$}
        \ncline[nodesep=0pt]{->}{ta}{p1}
        \Aput[1pt]{$x$}
        \ncline[nodesep=0pt]{->}{p2}{ta}
        \Bput[1pt]{$y$}
        \ncline[nodesep=0pt]{->}{ta}{p3}
        \Bput[1pt]{$y$}
       \uput{0}[0](-1,0){\Huge$\rightsquigarrow$}
        \cnodeput(1,0.8){p0}{\scriptsize$a$}
        \cnodeput(5,0.8){p1}{\white\scriptsize$a$}
        \rput[B](3,0.7){\Rnode{ta}{\psframebox{$x=y$}}}
        \rput[B](3,-0.9){\Rnode{ta2}{\psframebox{$x\neq y$}}}
        \cnodeput(1,-0.8){p2}{\scriptsize$b$}
        \cnodeput(5,-0.8){p3}{\white\scriptsize$b$}
        \ncline[nodesep=0pt]{->}{p0}{ta}
        \Aput[2pt]{$x$}
        \ncline[nodesep=0pt]{->}{ta}{p1}
        \Aput[2pt]{$x$}
        \ncline[nodesep=0pt]{->}{p2}{ta}
        \aput[2pt](0.15){$x$}
        \ncline[nodesep=0pt]{->}{ta}{p3}
        \aput[1pt](0.85){$x$}
        \ncline[nodesep=0pt]{->}{p0}{ta2}
        \bput[2pt](0.15){$x$}
        \ncline[nodesep=0pt]{->}{ta2}{p1}
        \bput[1pt](0.85){$x$}
        \ncline[nodesep=0pt]{->}{p2}{ta2}
        \Bput[2pt]{$y$}
        \ncline[nodesep=0pt]{->}{ta2}{p3}
        \Bput[2pt]{$y$}
    }\end{center}
\vspace{0.5cm} \caption{The net on the left cannot check for inequalities (it can fire its transition when $a=b$ or $a\neq b$). The net on the right
can fire the transition in the top when $a=b$, and the one in the bottom when $a\neq b$.}\label{fig:igualVsNoIgual}\vspace{0cm}
\end{figure}

\section{Undecidability of reachability for $\nu$-PN}\label{sec:reach}

Let us now prove that reachability is undecidable for $\nu$-PN.
In~\cite{Kummer00} (and independently in~\cite{Secco}) undecidability
of reachability is proved by reducing reachability of the final state
with all the counters containing zero in Minsky machines to reachability in $\nu$-PN. 
In this section we prove that same result in a more simple way, by reducing reachability
of inhibitor nets (that allow to check for zero) to reachability in $\nu$-PN.

  An inhibitor net is a tuple
  $N=(P,T,F,F_{in})$, where $P$ and $T$ are disjoint sets of places and transitions, respectively,
  $F\subseteq (P\times T)\cup (T\times P)$, and
  $F_{in}\subseteq P\times T$.
Pairs in $F_{in}$ are inhibitor arcs. For a transition $t\in T$
we write $^\bullet t=\{p\in P\mid (p,t)\in F\}$, $t^\bullet=\{p\in P\mid (t,p)\in F\}$ and $^i t=\{p\in
P\mid (p,t)\in F_{in}\}$.
In figures we will draw a circle instead of an arrow to indicate that an arc
is an inhibitor arc.

  A \emph{marking} of an inhibitor net $N$ is a multiset of places of $N$. A transition $t$
  of $N$ is enabled if $M(p)>0$ for all $p\in {^\bullet t}$ and $M(p)=0$ for all $p\in {^i t}$.
  In that case $t$ can be fired, producing $M'=(M- {^\bullet t})+t^\bullet$.

\begin{proposition}
 Reachability is undecidable for $\nu$-PN.
\end{proposition}

\begin{proof}
 Given an inhibitor net $N=(P,T,F,F_{in})$ we build a \hbox{$\nu$-PN} $N^*=(P\cup\bar P,T,F^*)$ that simulates it as follows:
 \begin{itemize}
  \item If $(p,t)\in F$ then $F^*(p,t)=F^*(\bar p,t)=F^*(\bar p,t)=\{x_p\}$ 
        (and analogously for $(t,p)\in F)$,
  \item If $(p,t)\in F_{in}$ then $F^*(\bar p,t)=\{x_p\}$ and $F^*(t,\bar p)=\{\nu\}$.
  \item $F^*(x,y)=\emptyset$ elsewhere.
\end{itemize}

Moreover, if $M_0$ is the initial marking of $N$, we consider a different identifier
$a_p$ for each place $p$ of $N$. Then, we define the initial marking of $N^*$ as
$M^*_0(\bar p)=\{a_p\}$ and $M^*_0(p)=\{a_p,\overset{M_0(p)}\ldots,a_p\}$, for each $p\in P$.

 Intuitively, for each place $p$ of $N$ we consider a new place $\bar p$ in $N^*$. 
 The construction of $N^*$ is such that $\bar p$ contains a single
 token at any time. The firing
 of any transition ensures that the token being used in $p$ coincides with that
 in $\bar p$. Every time a transition checks the emptyness of a place $p$, the content
 of $\bar p$ is replaced by a fresh token, so that no token remaining in $p$ can be
 used. In this way, our simulation introduces some garbage tokens whenever it cheats, that once 
 become garbage, always stay like that. Moreover, notice that any marking of $N^*$ of the
 form $M^*$ for some marking $M$ of $N$ does not contain any garbage, so that it comes from
 a correct simulation. 
 Fig.~\ref{fig:inh} depicts a simple inhibitor net and its simulation.
Then $M$ is reachable in $N$ from $M_0$ if and only if $M^*$ is reachable in $N^*$ from $M^*_0$. 
Thus, we have reduced reachability in inhibitor nets, which is undecidable~\cite{Esparza94}, to reachability in $\nu$-PN.
\end{proof}

\begin{figure}[!t]
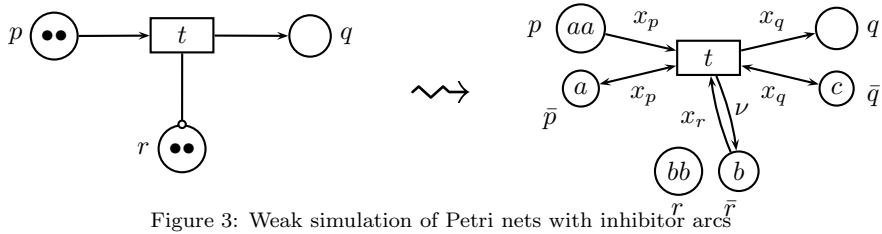

\vspace{0.5cm}
\begin{center}
 \scalebox{1}{
        \cnodeput(-5.2,0.3){p}{$\bullet\bullet$}
        \cnodeput(-3.5,-1.2){r}{$\bullet\bullet$}
        \cnodeput(-1.8,0.3){q}{\white$k$}
        \rput[B](-3.5,0.2){\Rnode{t}{\psframebox{$~~t~~$}}}
        \ncline[nodesep=0pt]{->}{p}{t}
        \ncline[nodesep=0pt]{o-}{r}{t}
        \ncline[nodesep=0pt]{->}{t}{q}
        \uput{0}[0](-5.8,0.3){$p$}
        \uput{0}[0](-1.4,0.3){$q$}
        \uput{0}[0](-4.1,-1.2){$r$}
	\uput{0}[0](-0.5,-0.5){\Huge$\rightsquigarrow$}
        \cnodeput(1.8,0.4){p}{$aa$}
        \cnodeput(1.8,-0.4){pn}{$a$}
        \cnodeput(3.1,-1.5){r}{$bb$}
        \cnodeput(3.9,-1.5){rn}{$b$}
        \cnodeput(5.2,0.4){q}{\white$k$}
        \cnodeput(5.2,-0.4){qn}{$c$}
        \rput[B](3.5,-0.1){\Rnode{t}{\psframebox{$~~t~~$}}}
        \ncline[nodesep=0pt]{->}{p}{t}
        \Aput[0.1]{$x_p$}
        \ncline[nodesep=0pt]{<->}{pn}{t}
        \Bput[0.1]{$x_p$}
        \ncarc[nodesep=0pt]{->}{rn}{t}
        \Aput[0.1]{$x_r$}
        \ncarc[nodesep=0pt]{->}{t}{rn}
        \Aput[1pt]{$\nu$}
        \ncline[nodesep=0pt]{->}{t}{q}
        \Aput[0.1]{$x_q$}
        \ncline[nodesep=0pt]{<->}{t}{qn}
        \Bput[0.1]{$x_q$}
        \uput{0}[0](1.1,0.4){$p$}
        \uput{0}[0](1.3,-0.8){$\bar p$}
        \uput{0}[0](3,-2){$r$}
	\uput{0}[0](3.7,-2){$\bar r$}
	\uput{0}[0](5.6,0.4){$q$}
	\uput{0}[0](5.6,-0.5){$\bar q$}
    }
 \end{center}
\vspace{1.2cm}\caption{Weak simulation of Petri nets with inhibitor arcs}\label{fig:inh}\vspace{-0.3cm}
\end{figure}

\section{Strict well structuredness of $\nu$-PN}\label{sec:coverability}

In this section we prove that the transition sytem generated by a $\nu$-PN
is strictly well structured~\cite{AbdullaCJT00,Finkel01Everywhere}. 
This will imply decidability of coverability, boundedness and termination.
For that purpose, we can proceed following the next steps.
In the first place, we need to
define an order in the set of configurations, markings in
our case, that induces the property of coverability. This order
must be a decidable wqo. Then
we must prove that this order is strictly monotonic with respect to the transition relation.
Finally, we have to prove that it has effective $Pred$-basis.

\subsection{Defining the order}

One could think that the order we are 
interested in for $\nu$-PN is the following:

\begin{center}
  $M_1\sqsubseteq M_2\Leftrightarrow M_1(p)\subseteq M_2(p)$ for all $p\in P$
\end{center}

This order is not a
well quasi-order: it suffices to consider a $\nu$-PN with a single place $p$ 
and a sequence of pairwise
different identifiers $(a_i)_{i=1}^\infty$, ant define 
$M_i(p)=\{a_i\}$ for all $i=1,2,\ldots$ which trivially 
satisfies that for all $i<j$, $M_i\not\sqsubseteq M_j$.

However, this order is too restrictive, since it does not take
into account the abstract nature of pure names. Indeed, 
whenever a new name is created, actually any other
fresh name could have been created. Therefore, reachability
(or coverability) of a given marking is equivalent to reachability
(or coverability) of any marking produced after consistently
renaming the new names in it. For homogeneity, we will suppose
that we can rename every name, even those appearing in the
initial marking (which, after all, are a fixed number of names).
To capture these intuitions, 
we identify markings up to renaming of names.

\begin{definition}
  Given two markings $M$ and $M'$ we say that they are \hbox{$\alpha$-equivalent,}
  and we write $M\equiv_\alpha M'$,
  if there is a bijection $\iota:\Id(M)\to\Id(M')$ such that
  $M'(p)(\iota(a))=M(p)(a)$ for all $p\in P$ and $a\in\Id(M)$.
\end{definition}

We will write $M\equiv_\iota M'$ to stress the use of the particular mapping
$\iota$ in the previous definition. Moreover, for a marking $M$ and
set of identifiers $A$, any bijection $\iota:\Id(M)\to A$ defines a
marking that we denote as $\iota(M)$, given by $\iota(M)(p)(\iota(a))=M(p)(a)$,
which is \hbox{$\alpha$-equivalent} to $M$.

\begin{proposition}\label{prop:alpha}
  The behavior of $\nu$-PNs is invariant under
  \hbox{$\alpha$-conversion.} More specifically, let $M_1\too{t(\sigma)}M'_1$:
  \begin{itemize}
   \item If $M_1\equiv_\alpha M_2$
  then there is $M'_2$ and $\sigma'$
  such that $M'_1\equiv_\alpha M'_2$ and $M_2\too{t(\sigma')}M'_2$.
   \item If $M'_1\equiv_\alpha M'_2$
  then there is $M_2$ and $\sigma'$
  such that $M_1\equiv_\alpha M_2$ and $M_2\too{t(\sigma')}M'_2$.
  \end{itemize}

\end{proposition}

\begin{proof} Let $A=\Id(M_1)\setminus\Id(M'_1)$ and $B$ the set of names created by $t(\sigma)$.
 Then, $\Id(M'_1)=(\Id(M_1)\setminus A)\cup B$. Notice that $B\subseteq\{b\}$, for some $b\in\Id$,
 assuming $N$ is normal.
  \begin{itemize}
   \item Assume $M_1\equiv_\iota M_2$ and let $\sigma'=\iota\circ\sigma$.
  Transition $t$ can be fired from $M_2$ with mode $\sigma'$,
  obtaining $M'_2$ with $\Id(M'_2)=(\Id(M'_1)\setminus\iota(A))\cup B'$ for some $B'$ of the same
  cardinality than $B$. We define $\iota'$ by extending $\iota$ to $B$ so that $\iota(B)=B'$, which verifies
  $M_2\equiv_{\iota'}M'_2$.
  \item Assume now that $M_1\equiv_{\iota'} M_2$ and let us define $\iota:\Id(M_1)\to\Id(M'_2)\cup A$ by
  $\iota(a)=\iota'(a)$ if $a\in\Id(M'_1)$, and $\iota(a)=a$ if $a\in A$. Then
  $M_2=\iota(M_1)$ and $\sigma'=\iota\circ\sigma$ satisfy $M_1\equiv_{\iota'}M_2$ and $M_2\too{t(\sigma')}M'_2$.
   \end{itemize}
\end{proof}

For instance, if we represent a marking $M$ of the net
in Fig.~\ref{fig:aut} by a tuple
$(M(p_1),M(p_2),M(p_3),M(p_4))$, then $M_1=(\{a,b\},\{b,c\},\emptyset,\emptyset)$ and
$M_2=(\{a,c\},\{b,c\},\emptyset,\emptyset)$ are two
\hbox{$\alpha$-equivalent} markings of that $\nu$-PN. 
Indeed, $M_1\equiv_\iota M_2$ with $\iota(a)=a$, $\iota(b)=c$
and $\iota(c)=b$.
$M_1$ can evolve to
the marking $M'_1=(\emptyset,\{c\},\{a,d\},\{d,e\})$ when it fires $t$ and 
$M_2$ can evolve to
$M'_2=(\emptyset,\{b\},\{a,e\},\{d,e\})$.
Notice that also $M'_1\equiv_\alpha M'_2$.

Let us now define the order we are interested in, by modifying the order $\sqsubseteq$ between markings with the help of the
\hbox{$\alpha$-equivalence} relation $\equiv_\alpha$.

\begin{definition}\label{def:orderalpha}
  Let $M_1$ and $M_2$ be markings of $N$. We will write \hbox{$M_1\sqsubseteq_\alpha M_2$}
  if there is a marking
  $M'_1$ such that $M'_1\equiv_\alpha M_1$ and $M'_1\sqsubseteq M_2$.
\end{definition}

Then, $M_1\sqsubseteq_\alpha M_2$ when there is $\iota$ such that
$M_1\equiv_\iota M'_1\sqsubseteq M_2$, or equivalently, when $\iota(M_1)\sqsubseteq M_2$. 
We will write $M_1\sqsubseteq_\iota M_2$ to emphasize on the use of 
$\iota$. Clearly, $\sqsubseteq_\alpha$ is a decidable quasi order. 
Moreover, the kernel
of $\sqsubseteq_\alpha$ is $\equiv_\alpha$, that is, 
$M_1\sqsubseteq_\alpha M_2$ and $M_2\sqsubseteq_\alpha M_1$ iff $M_1\equiv_\alpha M_2$.

\subsection{$\sqsubseteq_\alpha$ is a wqo}

We will now see that the set of markings, ordered by $\sqsubseteq_\alpha$, is a wqo.
In particular, notice that the 
counterexample we saw to prove that $\sqsubseteq$ is not a
wqo is no longer valid, since all those markings
are $\alpha$-equivalent.
In order to prove that $\sqsubseteq_\alpha$
is a wqo we map it to a multiset order
which is known to be a wqo.

A marking is a mapping $M:P\to\Id\to\nat$ that says, for a given
place $p$ and an identifier $a$, how many times the token $a$ 
can be found in place $p$. However, we can also currify those
mappings as $M:\Id\to P\to\nat$. Since 
the behavior of a net is invariant under renaming,
as we proved in Prop.~\ref{prop:alpha}, we can represent
markings (modulo $\equiv_\alpha$) as multisets in $\mMS{(P\to\nat)}$, that is, in $\mMS{(\mMS P)}$.

In this way, we represent markings by means of multisets,
with a cardinality that equals the number of different
identifiers appearing in it.

As an example, let us consider
a net with only two places $p_1$ and $p_2$, and a marking
$M$ such that $M(p_1)=\{a,a,b,c\}$ and
$M(p_2)=\{b,c\}$. We can represent that marking by
the multiset of cardinality 3, since there are 3
different identifiers in $M$, namely by the multiset
$\{\{p_1,p_1\},\{p_1,p_2\},\{p_1,p_2\}\}$, where the
multiset
$\{p_1,p_1\}$ represents identifier $a$, one of the
two multisets $\{p_1,p_2\}$ represents $b$ and the other
$\{p_1,p_2\}$ represents $c$.
Let us see it formally:

\begin{definition}\label{def:nf1}
   For a marking $M$ of $N$, we define 
   $M_a\in\mMS P$ by
   $M_a(p)=M(p)(a)$ and $\overline{M}=\{M_a\mid
   a\in\Id(M)\}\in\mMS{(\mMS P)}$. 
\end{definition}

Let us denote by $\ll$ the canonic order in $\mMS{(\mMS P)}$. It
is well known that $\ll$ is a wqo.
Moreover, it coincides with $\sqsubseteq_\alpha$,
as we prove next.

\begin{lemma}\label{equivorders}
    Let $M_1$ and $M_2$ be two markings. Then \hbox{$M_1\sqsubseteq_\alpha
M_2$} iff $\overline{M}_1\ll\overline{M}_2$.
\end{lemma}

\begin{proof}
 Let $\oM_1=\{A_1,\ldots,A_n\}$ and $\oM_2=\{B_1,\ldots,B_n\}$ with $A_i=M_i^{a_i}$ and
 $B_j=M_2^{b_j}$. If $M_1\sqsubseteq_\iota M_2$ then define $h(i)$ such that $B_{h(i)}=M_2^{\iota(a_i)}$.
 Then $A_i(p)=M_1(p)(a_i)\leq M_2(p)(\iota(a_i))=B_{h(i)}(p)$, so that $A_i\subseteq B_{h(i)}$ and
 therefore $\overline{M}_1\ll\overline{M}_2$.

 Conversely, since $\oM_1\ll\oM_2$, there is $h:\{1,\ldots,n\}\to\{1,\ldots,m\}$
 such that $A_i\subseteq B_{h(i)}$. Let us define $\iota:\Id(M_1)\to\Id(M_2)$ by
 $\iota(a_i)=b_{h(i)}$. Then we have $M_1(p)(a_i)=M_1^{a_i}(p)\leq M_2^{b_{h(i)}}(p)=M_2^{\iota(a_i)}(p)=M_2(p)(\iota(a_i))$.
 Therefore, $M_1(p)(a)\leq M_2(p)(\iota(a))$ for all $a\in\Id(M_1)$ and the thesis follows. 
\end{proof}

Finally, we can conclude that the order $\sqsubseteq_\alpha$ is,
indeed, a wqo.

\begin{proposition}
  $\sqsubseteq_\alpha$ is a wqo.
\end{proposition}

\begin{proof}
  Let $M_0,M_1,M_2,\ldots$ be an infinite sequence of markings. 
 Let us consider the sequence
 $\overline{M_0},\overline{M_1},\overline{M_2},\ldots$.
  Since
$\ll$ is a wqo, there are two indices $i<j$ such that
$\overline{M_i}\ll\overline{M_j}$. By Lemma~\ref{equivorders}
we have that $M_i\sqsubseteq_\alpha M_j$, from which the thesis
 follows.
\end{proof}

\subsection{Strict monotonicity}

Now let us see the next condition for strict well-structuredness, namely strict monotonicity
of the firing relation with respect to $\sqsubseteq_\alpha$. 
As a first step, let us see it for $\sqsubseteq$.

\begin{lemma}\label{lema:mono}
  The firing relation of $\nu$-PN is
  strictly monotonic with respect to~$\sqsubseteq$.
\end{lemma}

\begin{proof}
  Let us suppose that $M_1\too{t(\sigma)}M_2$ and $M_1\sqsubset
  M'_1$. From the former, we know in the first place that
  $\sigma(F(p,t))\in M_1(p)$ for all $p$ because that firing
  is enabled, and
  $M_2(p)=M_1(p)-\{\sigma(F(p,t))\}+\{\sigma(F(t,p))\}$ by definition
  of firing. The latter implies 
  $M_1(p)\subset M'_1(p)$. Then, for all $p$,
  $\sigma(F(p,t))\in M_1(p)\subseteq M'_1(p)$ and, therefore,
  the transition is enabled in $M'_1$. So that $t$ can be fired to obtain
  $M'_2(p)=M'_1(p)-\{\sigma(F(p,t))\}+\{\sigma(F(t,p))\}$. 
  Since $M_1(p)\subset M'_1(p)$ we have that
  $M_2(p)=M_1(p)-\{\sigma(F(p,t))\}+\{\sigma(F(t,p))\}\subset
  M'_1(p)-\{\sigma(F(p,t))\}+\{\sigma(F(t,p))\}=M'_2(p)$ and the thesis
  follows.
\end{proof}

\begin{proposition}\label{prop:monotonic_alpha}
  The firing relation of $\nu$-PN is strictly
  monotonic with respect to~$\sqsubseteq_\alpha$.
\end{proposition}

\begin{proof}
  It is a direct consequence of the previous lemma and Prop.~\ref{prop:alpha}.
\end{proof}

\subsection{Effective $\Pred$-basis}

Let us now move to the last condition we must check, effective $\Pred$-basis.
Let us denote by ${\uc M}$ and $\uca M$ the upward closure of $M$ with respect
$\sqsubseteq$ and $\sqsubseteq_\alpha$, respectively.

\begin{definition}
    Given a transition $t$ of $N$ and $\sigma$ a mode for $t$, we define 
    $\Pred_t$ and $Pre_{t(\sigma)}$ as the functions mapping markings to sets
   of markings, defined by
    $\Pred_t(M)=\{M'\mid\exists\sigma~M'\too{t(\sigma)}M\}$ and $\Pred_{t(\sigma)}(M)=\{M'\mid M'\too{t(\sigma)}M\}$, and
    extend them pointwise to sets of markings.
\end{definition}

With these notations we need to compute $\minn(\Pred_t(\uca M))$ for each marking $M$ and $t\in T$.
By Prop.~\ref{prop:alpha} it is enough to compute $\minn(\Pred_t(\uc M))$. Notice that 
the minor set of $\Pred_t(\uc M)$ is still considered with respect to $\sqsubseteq_\alpha$, so that it
is finite.

When computing the predecessors of $\uparrow M$, it may be the case
that $M$ itself has no predecessors, but some other markings in
$\uparrow M$ do. In the next definition we identify the least marking in
$\uc M$ with predecessors. We will use the following notation:
Given two markings $M_1$ and $M_2$ we will denote by
$M_1\sqcup M_2$ the marking given by $(M_1\sqcup
M_2)(p)=M_1(p)\sqcup M_2(p)$.

\begin{definition}
  Let $t$ be a transition of $N$, $\sigma$ a mode of $t$ and $M$ a marking of $N$. We define
 $\minn_{t(\sigma)}(M)=M\sqcup\sigma(F(t,-))$, where $\sigma(F(t,-))$ is the marking of $N$ defined by
$\sigma(F(t,-))(p)=\sigma(F(t,p))$.
\end{definition}

Indeed, $\minn_{t(\sigma)}(M)$ is a marking in $\uparrow M$ with some predecessors.
Moreover, is the least such marking, as proved next.

\begin{lemma}\label{lemma:min}
   Let $M$ be a marking of $N$, $t$ a transition of $N$ and
 $\sigma$ a mode of $t$.
   Then $\minn_{t(\sigma)}(M)$ is the least $M'$ such that
   $M\sqsubseteq M'$ and
$Pre_{t(\sigma)}(M')\neq\emptyset$.
\end{lemma}

\begin{proof} Let us write $\bar{M}=\minn_{t(\sigma)}(M)$. Trivially, $M\sqsubseteq\bar{M}$.
Let us see that $Pre_{t(\sigma)}(\bar{M})\neq\emptyset$. For that
purpose, let
$M_0$ be the marking defined by
$M_0(p)=(M\sqcup\sigma(F(t,-)))(p)-\{\sigma(F(t,p))\}+\{\sigma(F(p,t))\}$
and let us see that $M_0\too{t(\sigma)}\bar{M}$. In the first place,
$t(\sigma)$ is enabled in $M_0$, since $\sigma(F(p,t))\in
M_0(p)$ for each place $p$. Then the transition can be fired in mode
 $\sigma$ and
$M_0(p)-\sigma(F(p,t))+\sigma(F(t,p))=\bar{M}(p)$.
Finally, if $M_1\too{t(\sigma)}M_2$ and $M\sqsubseteq M_2$ let us see that $\bar{M}\sqsubseteq M_2$.
Since $M\sqsubseteq M_2$ it holds
that $M(p)\subseteq M_2(p)$, for all $p$. Then
$\bar{M}(p)=M(p)\sqcup\sigma(F(t,p))\subseteq
M_2(p)\sqcup\sigma(F(t,p))\subseteq M_2(p)$, and the thesis follows.
\end{proof}

Finally, let us see that we can use $\minn_{t(\sigma)}(M)$ to compute
$\minn(\Pred_t(\uparrow M))$.

\begin{figure}[!t]
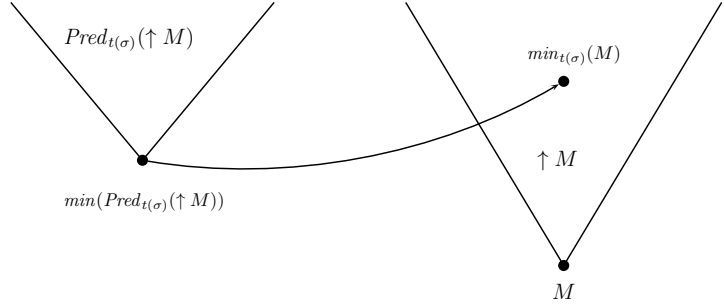

  \begin{center}
  \vspace{3cm}
  \scalebox{0.7}{
  \psline(-6.5,5)(-4,2)(-1.5,5)
   \psline(1,5)(4,0)(7,5)
\cnode*(4,0){3pt}{M}\cnodeput*(4,-0.5){M}{\large$M$}
\cnode*(4,3.5){3pt}{min}
    \uput{0}[0](3.3,4){$\minn_{t(\sigma)}(M)$}
\uput{0}[0](-5.5,1.2){$\minn(\Pred_{t(\sigma)}(\uparrow M))$}
\cnode*(-4,2){3pt}{premin1} 
    \uput{0}[0](3.5,2){\large$\uparrow M$}
    \uput{0}[0](-5.5,4.3){\large$\Pred_{t(\sigma)}(\uparrow M)$}
    \ncarc[nodesep=0pt,arcangle=-20]{->}{premin1}{min}
  }\vspace{0cm}
\end{center}
\caption{Computation of $\Pred_{t(\sigma)}(\uparrow M)$}\label{fig:ideales}
\end{figure}

\begin{proposition}\label{prop:pred-min}
 $\Pred_{t(\sigma)}(\uc M)={\uc\Pred_{t(\sigma)}}(\minn_{t(\sigma)}(M))$
\end{proposition}

\begin{proof}
 Let $\oM$ such that $\Pred_{t(\sigma)}({\uc M})={\uc\oM}$. Since $\minn_{t(\sigma)}(M)\in{\uc M}$,
 $\Pred_{t(\sigma)}(\minn_{t(\sigma)}(M))\in{\uc\oM}$, so that $\oM\sqsubseteq\Pred_{t(\sigma)}(\minn_{t(\sigma)}(M))$.
 Let us see that also $\Pred_{t(\sigma)}(\minn_{t(\sigma)}(M))\sqsubseteq\oM$ holds.
 Indeed, $\oM\in\Pred_{t(\sigma)}({\uc M})$, so there is $M'\in{\uc M}$ such that $\oM\too{t(\sigma)}M'$.
 By the previous lemma, since $M'$ has predecessores, $\minn_{t(\sigma)}(M)\sqsubseteq M'$, which entails by the previous lemma
 that the same relation also holds for their predecessors (because the effect of $t(\sigma)$ is constant), and hence the thesis.
\end{proof}

Fig.~\ref{fig:ideales} can give you some insight about the proof
of the previous result. A marking $M$ induces an upwards
closed set, the cone in the right handside of Fig.~\ref{fig:ideales}.
We want to compute (a finite representation of) the set of the predecessors of the 
markings in that cone. For that purpose, we first obtain $\minn_{t(\sigma)}(M)$, which
is known to have a predecessor, according to Prop.~\ref{lemma:min}, that is trivially computable.
Therefore, every marking in the left handside cone can reach in one step the cone in the right.

Let us now see that in order to compute $\minn(\Pred_t(\uc M))$ it is enough to consider
a finite ammount of modes.

\begin{proposition}\label{prop:finite-modes}
 Let $M$ be a marking, $t$ a transition and $O$ a set of identifiers with $|O|=|\Var(t)|$.
 If $M'\in\Pred_t(\uc M)$ then there is $\sigma:\Var(t)\to\Id(M)\cup O$ and $M''\equiv_\alpha M'$
 such that $M''\in\Pred_{t(\sigma)}(\uc M)$.
\end{proposition}

\begin{proof}
 Let $\sigma'$ such that $M'\too{t(\sigma')}\oM$ with $M\sqsubseteq\oM$. Because of the latter,
 $\Id(\oM)=\Id(M)\cup O'$ for some set of identifiers $O'$. Let us write $\sigma'(x)=o'_x$ whenever
 $\sigma'(x)\in O'$. For each such $x\in\Var(t)$, choose a different $o_x\in O$ (notice that this
 can be done because $|O|=|\Var(t)|$). Let us define $\sigma:\Var(t)\to\Id(M)\cup O$ as follows:
 $\sigma(x)=\sigma'(x)$ if $\sigma'(x)\in\Id(M)$, and $\sigma(x)=o_x$ if $\sigma'(x)\in O'$.
 Let also $\iota:\Id(M')\to(\Id(\oM)\setminus O')\cup O$ defined by $\iota(o'_x)=o_x$ and
 $\iota(a)=a$ elsewhere. Finally, let us take $M''=\iota(M')$ and $\oM'$ such that
 $M''\too{t(\sigma)}\oM'$. It holds that $\oM'\in{\uc M}$ and the thesis follows.
\end{proof}

Therefore, in order to compute $\Pred_t(\uc M)$ we can fix a set of names $O$ with as many
names as variables in $\Var(t)$, and consider only modes mapping variables to names
in $\Id(M)$ or in $O$. Notice that there are finitely many such modes.

\begin{proposition}
  For each $M$, the set $\mathit{\minn(\Pred_t(\uparrow M))}$ is
  computable.
\end{proposition}

\begin{proof}
    We can compute $\minn(\Pred_t(\uparrow M))$ as follows:
\[\minn(\Pred_t(\uparrow M))=\minn\big(\underset{\sigma}{\bigcup}\Pred_{t(\sigma)}(\uparrow M)\big)=\minn\Big(\underset{\sigma}{\bigcup}\minn(\Pred_{t(\sigma)}(\uparrow M))\Big)\]
By Prop.~\ref{prop:pred-min} the last term can be computed as
$\minn\big(\underset{\sigma}{\bigcup}\Pred_{t(\sigma)}(\minn_{t(\sigma)}(M))\big)$.
Each $\Pred_{t(\sigma)}(\minn_{t(\sigma)}(M))\big)$ is computable, 
and because by Prop.~\ref{prop:finite-modes} it is enough to consider
finitely many modes, we conclude.
\end{proof}

We have proved that $\nu$-PNs are strictly well structured transitions systems.

\begin{proposition}\label{prop:cbt}
  Coverability, boundedness and termination are decidable for $\nu$-PN.
\end{proposition}

One can think that we have proved decidability of a weak version
of the coverability problem, that in which we allow arbitrary
renaming of identifiers. For instance, if we consider the net in the left of
Fig.~\ref{fig:igualVsNoIgual},
 and we ask whether the marking $M$ given by
$M(p_0)=M(p_1)=\emptyset$, $M(p_3)=\{b\}$ and $M(p_4)=\{a\}$
can be covered, the result would be affirmative, since
the marking obtained by exchanging $a$ and $b$ in $M$ (which is
$\alpha$-equivalent to $M$) is reachable in one step. 

However, we can use this apparently
weak version to decide a more restricted version of coverability:
Let $M_0$ and $M_f$ be two markings of a $\nu$-PN
$N=(P,T,F)$. We want to decide if we can cover $M_f$ from $M_0$ 
without allowing renaming of names. Thus, if a name $a$ appears
both in $M_0$ and in $M_f$ we want to reach a marking $M$
such that $M_f\sqsubseteq_\iota M$ with $\iota$ satisfying
$\iota(a)=a$. Since $R=\Id(M_0)\cap\Id(M_f)$ contains only a finite 
number of names, we can add new places in order to ensure the latter.
We define the $\nu$-PN $N^*=(P\cup R,T,F)$. For any marking
$M$ we define $M^*(p)=M(p)$ if $p\notin R$ and $M^*(r)=\{r\}$ for
all $r\in R$. By construction of $N^*$, places in 
$R$ are isolated, so that their tokens are never moved or removed.
In particular, for any reachable $M$ with $M_f\sqsubseteq_\iota M$ it
holds $\iota(a)=a$ for every $a\in R$.

\begin{figure}[!t]
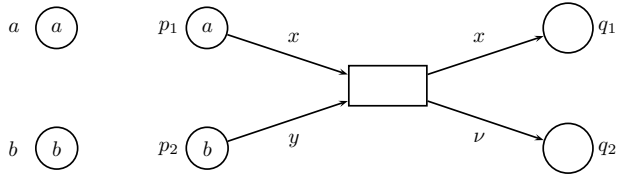

 \vspace{1cm}
 \begin{center}
 \scalebox{0.8}{
  \cnodeput(-4.5,1){pa}{$~a~$}
    \cnodeput(-4.5,-1){pb}{$~b~$}
  \cnodeput(-2,1){p1}{$~a~$}
    \cnodeput(-2,-1){p2}{$~b~$}
    \cnodeput(4,1){p3}{\white$~\bullet~$}
    \cnodeput(4,-1){p4}{\white$~\bullet~$}
    \rput[B](1,-0.1){\Rnode{t}{\psframebox{\Large\white$sas_1!$}}}
     \ncline[nodesep=0pt]{->}{p1}{t}
     \Aput{$x$}
    \ncline[nodesep=0pt]{->}{p2}{t}
    \Bput{$y$}
     \ncline[nodesep=0pt]{->}{t}{p3}
    \Aput{$x$}
\ncline[nodesep=0pt]{->}{t}{p4}
  \Bput{$\nu$}
  \uput{0}[0](-2.8,1){$p_1$}
        \uput{0}[0](-2.8,-1){$p_2$}
        \uput{0}[0](4.5,1){$q_1$}
        \uput{0}[0](4.5,-1){$q_2$}
        \uput{0}[0](-5.3,1){$a$}
        \uput{0}[0](-5.3,-1){$b$}}\vspace{1.5cm}
\end{center}\vspace{-1cm}
\caption{The $\nu$-PN in the left of Fig.~\ref{fig:igualVsNoIgual} extended to decide a restricted version of coverability}\label{ejTonto2}
\vspace{-0cm}\end{figure}

Let us again consider the example in Fig.~\ref{fig:igualVsNoIgual}. 
Following the previous construction, that can be
seen in Fig.~\ref{ejTonto2}, we add a place for $a$ and another
one for $b$. When we execute this new net, the reasoning we followed
before now fails. In one step we can reach $M'$
with $M'(p_0)=M'(p_1)=\emptyset$, $M'(p_3)=M'(a)=\{a\}$ and
$M'(p_4)=M'(b)=\{b\}$. However, thanks to the newly added places,
it is not true that $M'$ equals the result of exchanging $a$ and $b$ in $M^*$ (using the notations of the proof
of the previous result).

We could ask ourselves whether we can consider a ligther version of
the reachability problem  
in which we allow renaming of names, as we
are doing with coverability, that allows us to obtain decidability.
However, decidability of $\alpha$-reachability implies the
decidability of reachability, by using the same trick we have
used for coverability.

\subsection{Complexity of the decision procedures}

Now we obtain hardness results for the decision problems shown to be 
decidable in Prop.~\ref{prop:cbt}. We do it by means of a simulation
of reset nets by $\nu$-PN.
The construction
is very similar to the one we used in Sect.~\ref{sec:reach} to simulate
inhibitor nets with $\nu$-PN.

  A reset net is a tuple
  $N=(P,T,F,F_r)$, where $P$ and $T$ are disjoint sets of places and transitions, respectively,
  $F\subseteq (P\times T)\cup (T\times P)$, and
  $F_r\subseteq P\times T$.
Pairs in $F_r$ are reset arcs. For a transition $t\in T$
we write $^\bullet t=\{p\in P\mid (p,t)\in F\}$ and $^r t=\{p\in
P\mid (p,t)\in F_r\}$, and analogously for $t^\bullet$.
For simplicity, and without loss of generality, we assume that
${^r t}\cap t^\bullet=\emptyset$ for every $t\in T$.

  A marking of a reset net $N$ is a multiset of places of $N$. A transition $t$
  is enabled in $M$ if $M(p)>0$ for all $p\in {^\bullet t}$.
  In that case $t$ can be fired, producing $M'$ defined as\footnote{Note that we are identifying $F$ with its characteristic funcion.}
  \begin{itemize}
    \item $M'(p)=(M(p)-F(p,t))+F(t,p)$ for all $p\notin {^r t}$,
    \item $M'(p)=0$ for all $p\in {^r t}$.
  \end{itemize}

\begin{figure}[!t]
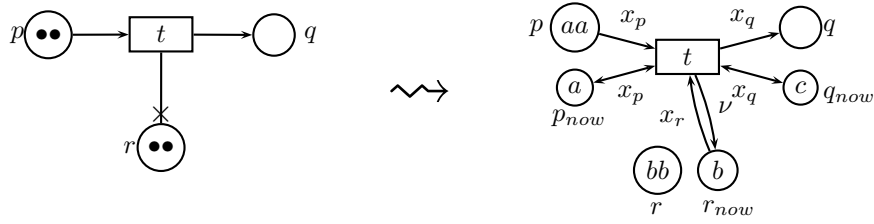

\vspace{0.3cm}
\begin{center}
 \scalebox{1}{
        \cnodeput(-5,0.3){p}{$\bullet\bullet$}
        \cnodeput(-3.5,-1.2){r}{$\bullet\bullet$}
        \cnodeput(-2,0.3){q}{\white$k$}
        \rput[B](-3.5,0.2){\Rnode{t}{\psframebox{$~~t~~$}}}
        \ncline[nodesep=0pt]{->}{p}{t}
        \ncline[nodesep=0pt]{-}{r}{t}
        \ncline[nodesep=0pt]{->}{t}{q}
        \uput{0}[0](-5.5,0.3){$p$}
        \uput{0}[0](-1.6,0.3){$q$}
        \uput{0}[0](-4,-1.2){$r$}
	\uput{0}[0](-0.5,-0.5){\Huge$\rightsquigarrow$}
        \cnodeput(2,0.4){p}{$aa$}
        \cnodeput(2,-0.4){pn}{$a$}
        \cnodeput(3.1,-1.5){r}{$bb$}
        \cnodeput(3.9,-1.5){rn}{$b$}
        \cnodeput(5,0.4){q}{\white$k$}
        \cnodeput(5,-0.4){qn}{$c$}
        \rput[B](3.5,-0.1){\Rnode{t}{\psframebox{$~~t~~$}}}
        \ncline[nodesep=0pt]{->}{p}{t}
        \Aput[0.1]{$x_p$}
        \ncline[nodesep=0pt]{<->}{pn}{t}
        \Bput[0.1]{$x_p$}
        \ncarc[nodesep=0pt]{->}{rn}{t}
        \Aput[0.1]{$x_r$}
        \ncarc[nodesep=0pt]{->}{t}{rn}
        \Aput[0.1]{$\nu$}
        \ncline[nodesep=0pt]{->}{t}{q}
        \Aput[0.1]{$x_q$}
        \ncline[nodesep=0pt]{<->}{t}{qn}
        \Bput[0.1]{$x_q$}
        \uput{0}[0](1.4,0.4){$p$}
        \uput{0}[0](1.7,-0.8){$p_{now}$}
        \uput{0}[0](3,-2){$r$}
	\uput{0}[0](3.7,-2){$r_{now}$}
	\uput{0}[0](5.3,0.4){$q$}
	\uput{0}[0](5.3,-0.5){$q_{now}$}
	\uput{0}[0](-3.67,-0.75){\large$\times$}
    }
 \end{center}
\vspace{1.5cm}\caption{Simulation of reset nets}\label{fig:reset}\vspace{-0.3cm}
\end{figure}

\begin{proposition}\label{prop:nuPN-reset}
 Given a reset net $N=(P,T,F,F_r,M_0)$ we can build in polynomial time a $\nu$-PN $N^*=(P\cup\bar P,T,F^*,M^*_0)$
 such that:
   \begin{itemize}
     \item If $M$ is reachable in $N$ then there is $M^*$ reachable in $N^*$ such that for 
      every $p\in P$ there is $a_p\in\Id$ with $M^*(\bar p)=\{a_p\}$ and $M^*(p)(a_p)=M(p)$.
     \item If $M^*$ is reachable in $N^*$ then there is $M$ reachable in $N$ and $a_p\in\Id$ for every $p\in P$ such that
     $M^*(\bar p)=\{a_p\}$ and $M^*(p)(a_p)=M(p)$.
 \end{itemize}
In particular, 
 \begin{itemize}
  \item $N$ terminates iff $N^*$ terminates,
  \item Given $M$ we can also build $M^*$ such that $M$ can be covered in $N$ iff $M^*$ can be covered in $N^*$.
 \end{itemize}

\end{proposition}

\begin{proof}
 Let $N=(P,T,F,F_r)$ be a reset net.
 We consider a different variable $x_p$ for each $p\in P$.
 Then we define $N^*=(P\cup\bar P,T,F^*)$ as follows:
 \begin{itemize}
  \item If $(p,t)\in F$ then $F^*(p,t)=F^*(\bar p,t)=F^*(\bar p,t)=x_p$ (analogously for $(t,p)\in F)$,
  \item If $(p,t)\in F_r$ then $F^*(\bar p,t)=x_p$ and $F^*(t,\bar p)=\nu$.
  \item $F^*(x,y)=\emptyset$ elsewhere.
\end{itemize}

Moreover, if $M_0$ is the initial marking of $N$, we consider a different identifier
$a_p$ for each place $p$ of $N$. Then, we define the initial marking of $N^*$ as
$M^*_0(p_{now})=\{a_p\}$ and $M^*_0(p)=\{a_p,\overset{M_0(p)}\ldots,a_p\}$, for each $p\in P$.

 Intuitively, for each place $p$ of $N$ we consider a new place $\bar p$ in $N^*$. 
 The construction of $N^*$ is such that $\bar p$ contains a single
 token at any time. The firing
 of any transition ensures that the token being used in $p$ coincides with that
 in $\bar p$. Every time a transition resets a place $p$, the content
 of $\bar p$ is replaced by a fresh token, so that no token remaining in $p$ can be
 used. In this way, our simulation introduces some garbage tokens, that once 
 become garbage, always stay like that. Fig.~\ref{fig:reset} depicts
 a simple reset net and its simulation. 
\end{proof}

\begin{proposition}
 Coverability, boundedness and termination for $\nu$-PN are not primitive recursive.
\end{proposition}

\begin{proof}
 Since coverability and termination are Ackermand-hard for reset nets~\cite{Sch10},
 the previous construction entails Ackerman-hardness for coverability and
 termination in $\nu$-PN. This hardness extends to boundedness by means of a
 very simple reduction: given a $\nu$-PN $N$ it is enough to build $N'$ by adding to $N$ a place in which
 an ordinary token is put in every firing. Clearly, $N$ terminates iff $N'$ is bounded.
\end{proof}

\section{Weaker forms of boundedness}\label{sec:bounded}

Let us now discuss weaker forms of boundedness.
In the first place, we characterize boundedness (finiteness of the
reachability set) in terms of the form of every reachable marking,
as is usual in Petri nets.

\begin{lemma}
  Given a $\nu$-PN with an initial marking, the set of reachable markings
  is finite (up to $\equiv_\alpha$) if and only if there is $n\geq0$ such that
  every reachable marking $M$ satisfies $M(p)(a)\leq n$ for all $p\in P$ and $a\in\Id$.
\end{lemma}

\begin{proof}
  If $\mathit{Reach}$ is finite we can 
  define $s=max\{|\Id(M)|\mid M\in \mathit{Reach}\}$ and $k=max\{M(p)(a)\mid M\in \mathit{Reach},~p\in P,~a\in\Id(M)\}$. 
  Then, for each reachable $M$, $|M(p)|=|\underset{a\in\supp(M(p))}{\sum}M(p)(a)|\leq k\cdot s$ and the
  net is bounded. Conversely, if the net is unbounded then for each $n$ there is
  a reachable $M_n$ such that $|M_n(p)|>n$ for all $p$, which implies the thesis.
\end{proof}

We will use the previous characterization in order
to factorize the property of boundedness.
Unlike ordinary P/T nets, that only have one infinite dimension,
$\nu$-PNs have two different sources of infinity: the number
of different identifiers and the number of times each of those
identifiers appear. Consequently, several different notions of
boundedness arise, in one of the dimensions, in the other or
in both.

\begin{definition}\label{def:acotacion} Let $N$ be a $\nu$-PN.
  \begin{itemize}
    \item We say $N$ is \emph{width-bounded} if there is $n\in\nat$ such that for all reachable $M$, 
    $|\Id(M)|\leq n$.
    \item We say $N$ is \emph{depth-bounded} if there is $n\in\nat$ such that for all reachable $M$,
    for all $p\in P$ and for all $a\in\Id, M(p)(a)\leq n$.
  \end{itemize}
\end{definition}

Indeed, width and depth-boundedness factorize boundedness.

\begin{proposition}\label{prop:acot_char}
  $N$ is bounded iff it is width-bounded and depth-bounded.
\end{proposition}

 \begin{proof}
     It is enough to consider that $|M(p)|=|\underset{a\in\Id(M)}{\sum}M(p)(a)|\leq |\Id(M)|\cdot max\{M(p)(a)\mid a\in\Id\}$. 
If there is $n\in\nat$ such that $|M(p)|\leq
    n$ then $|\underset{a\in\Id(M)}{\sum}M(p)(a)|\leq n$ and since 
    $\Id(M)=\{a\in\Id\mid M(p)(a)>0$ for some $p\}$ we have that $|\supp(M(p))|\leq n$ and
    and for all $a\in\supp(M(p))$, $M(p)(a)\leq n$. Conversely,
    let us assume there are $n$ and $m$ such that $|\supp(M(p))|\leq n$ and
    $M(p)(a)\leq m$. From the latter if follows that $max\{M(p)(a)\mid a\in\supp(M(p))\}\leq m$. 
    Then, by the previous observation,
    $|M(p)|\leq n\cdot m$ and the thesis follows.
 \end{proof}

Thanks to the previous result  
we know that if a $\nu$-PN is bounded then it is width-bounded
and depth-bounded. However, if it is unbounded it could still be
the case that it is width-bounded (see left of Fig.~\ref{fig:2no3}) or
depth-bounded (see right of Fig.~\ref{fig:2no3}), though not simultaneously
width and depth-bounded.

\begin{figure}[!t]
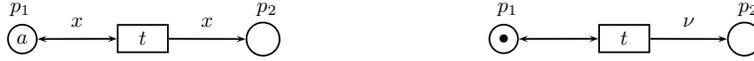

\vspace{0.8cm}
\begin{center}
 \scalebox{0.8}{
        \cnodeput(-6,0){p}{$a$}
        \cnodeput(-2,0){q}{\white$k$}
        \rput[B](-4,-0.1){\Rnode{t}{\psframebox{$~~t~~$}}}
        \ncline[nodesep=0pt]{<->}{p}{t}
        \Aput{$x$}
        \ncline[nodesep=0pt]{->}{t}{q}
        \Aput{$x$}
        \uput{0}[0](-6.2,0.5){$p_1$}
	\uput{0}[0](-2.1,0.5){$p_2$}
        \cnodeput(2,0){p}{$\bullet$}
        \cnodeput(6,0){q}{\white$k$}
        \rput[B](4,-0.1){\Rnode{t}{\psframebox{$~~t~~$}}}
        \ncline[nodesep=0pt]{<->}{p}{t}
        \ncline[nodesep=0pt]{->}{t}{q}
        \Aput{$\nu$}
	\uput{0}[0](1.9,0.5){$p_1$}
	\uput{0}[0](5.9,0.5){$p_2$}
    }
 \end{center}
\vspace{-0.2cm} \caption{Width-bounded but not depth-bounded $\nu$-PN (left) and viceversa (right)}\label{fig:2no3}
\vspace{-0cm}\end{figure}

In~\cite{atpn10} we prove decidability of width-boundedness for $\nu$-PN.
The proof relies on the results in~\cite{forwardI,forwardII} that establish
a framework for forward analysis for WSTS. We do not show the details here,
since they are rather involved. 

Though width and depth-boundedness seem to play a dual role, the proof
of decidability of width-boundedness can not be adapted in the case of
depth-boundedness. Actually, depth-boundedness turns out to be
undecidable, though this fact could be considered to be rather
anti intuitive (actually, in the paper~\cite{Dietze07} there is a wrong
decidability proof).

\begin{proposition}
 Depth-boundedness is undecidable for $\nu$-PN.
\end{proposition}

\begin{proof}
 Given a $\nu$-PN $N$, let us consider the reset net $N^*$ built in
 Prop.~\ref{prop:nuPN-reset}. Notice that $N$ is bounded iff $N^*$ is depth-bounded.
 Since boundedness in reset nets is undecidable~\cite{dufourd98reset} we can conclude.
\end{proof}

\section{Conclusions and Future Work}\label{sec:conclusion}

In this paper we have studied the expressive power of a simple extension
of P/T nets with a primitive that creates fresh names. We knew that the
expressive power of P/T nets is strictly increased because, unlike for P/T nets, 
reachability is undecidable. However, Turing-completeness is not reached. We have seen it
by proving that $\nu$-PNs are strictly well-structured systems. In particular, we obtain
that coverability is still decidable for them, as well as boundedness. Therefore,
$\nu$-PN is in the class of models whose expressive power lies somewhere in
between P/T nets and Turing machines, like Lossy FIFO channel systems~\cite{AbdullaJ93} or
reset nets~\cite{dufourd98reset}.

We have also defined two orthogonal notions of boundedness. Since our nets
have names as tokens, it can be the case that a bounded number of different
names appear in every reachable marking. In that case (independently of
the number of times that those each of those names appears) we say the net
is width-bounded. Dually, if every name that appears in every reachable
marking appears only a bounded number of times (independently of how many
different names appear) then we say that the net is depth-bounded.
Though width-boundedness is decidable, we have proved undecidability of 
depth-boundedness by reducing boundedness in reset nets to it.

Many well structured transition systems have undecidable reachability, except
some notable exceptions. Moreover, we know that coverability is always
decidable for them. Thus, in order to compare the expressive power of different formalisms that
lie in this class, reachability and coverability are not enough. One could
consider other properties, as different notions of boundedness, though we 
have seen that boundedness properties tend to be rather tricky. A different
option is to consider the languages generated when we label transitions with
labels taken from a finite set. Because of the undecidability of reachability,
if we accept words that can be recognized when reaching a given marking, then
we generally obtain the set of recursibly enumerable languages. In~\cite{WSL}
the authors propose to use coverability as accepting condition instead.
This yields a better framework to relate well structured transition systems.
In~\cite{LATA10} such framework is used in order to compare $\nu$-PN with
other Petri net extensions as Affine Well Nets or Data Nets. However, the
distinction between $\nu$-PN and Data Nets remains an open problem.

\bibliographystyle{elsarticle-num}

\end{document}